\documentclass[11pt]{article}
\usepackage[IL2]{fontenc}
\usepackage[letterpaper,margin=1.00in]{geometry}
\usepackage{amsmath, amssymb, amsthm, amsfonts}
\usepackage{bbm}
\usepackage{amscd}
\usepackage{mathrsfs}

\usepackage{comment} 
\usepackage{ifthen}
\usepackage{tikz}
\usetikzlibrary{positioning,decorations.pathreplacing}
\usepackage{ bbold }
\usepackage{appendix}
\usepackage{graphicx}
\usepackage{color}
\usepackage{epstopdf}
\usepackage{wrapfig}
\usepackage{paralist}
\usepackage{wasysym}
\usepackage[textsize=tiny]{todonotes}

\usepackage{listings} 

\usepackage{pdflscape} 

\usepackage{framed}
\usepackage[framemethod=tikz]{mdframed}
\usepackage[bottom]{footmisc}
\usepackage{enumitem}
\setitemize{noitemsep,topsep=3pt,parsep=3pt,partopsep=3pt}
\usepackage[font=small]{caption}
\usepackage{xspace}

\usepackage{thmtools} 
\usepackage{thm-restate} 

\usepackage{hyperref}
\hypersetup{
    unicode=false,          
    colorlinks=true,        
    linkcolor=red,          
    citecolor=darkgreen,        
    filecolor=magenta,      
    urlcolor=cyan           
}

\newtheorem{theorem}{Theorem}[section]
\newtheorem{lemma}[theorem]{Lemma}
\newtheorem{meta-theorem}[theorem]{Meta-Theorem}

\newtheorem{corollary}[theorem]{Corollary}

\newtheorem{observation}[theorem]{Observation}
\newtheorem{definition}[theorem]{Definition}

\usepackage[capitalize, nameinlink,noabbrev]{cleveref}

\crefname{theorem}{Theorem}{Theorems}
\crefname{proposition}{Proposition}{Propositions}
\crefname{observation}{Observation}{Observations}
\crefname{lemma}{Lemma}{Lemmas}
\crefname{claim}{Claim}{Claims}
\crefname{problem}{Problem}{Problems}
\crefname{conjecture}{Conjecture}{Conjectures}
\crefname{question}{Question}{Questions}
\crefname{example}{Example}{Examples}
\crefname{fact}{Fact}{Facts}

\definecolor{darkgreen}{rgb}{0,0.5,0}

\usepackage{algcompatible}
\algnewcommand\algorithmicswitch{\textbf{switch}}
\algnewcommand\algorithmiccase{\textbf{case}}

\algdef{SE}[SWITCH]{Switch}{EndSwitch}[1]{\algorithmicswitch\ #1\ \algorithmicdo}{\algorithmicend\ \algorithmicswitch}%
\algdef{SE}[CASE]{Case}{EndCase}[1]{\algorithmiccase\ #1}{\algorithmicend\ \algorithmiccase}%
\algtext*{EndSwitch}%
\algtext*{EndCase}%

\renewcommand{\phi}{\varphi}

\renewcommand{\paragraph}[1]{\vspace{0.15cm}\noindent {\bf #1}:}

\newcommand{\FullOrShort}{full}
\ifthenelse{\equal{\FullOrShort}{full}}{
  
  \newcommand{\fullOnly}[1]{#1}
  \newcommand{\shortOnly}[1]{}

  }{

    \newcommand{\fullOnly}[1]{}
    \newcommand{\IncludePictures}[1]{}
   
  }

\newcommand{\Pot}

\usepackage{algorithm}
\usepackage[noend]{algpseudocode}

\newcommand{\iapsp}{\textsc{Incorrect-APSP} }
\newcommand{\apsp}{\textsc{APSP}}

\usepackage{amsmath}
\hypersetup{
    colorlinks=true,
    filecolor=magenta,          urlcolor=black,
    pdftitle={Overleaf Example},
    pdfpagemode=FullScreen,
    }

\title{Anarchy in the APSP: Algorithm and Hardness for Incorrect Implementation of Floyd-Warshall}

\begin{document}
\date{}
 \author{Jaehyun Koo \\ \small MIT \\ \small koosaga@mit.edu}
\maketitle

\begin{abstract}
The celebrated Floyd-Warshall algorithm efficiently computes the all-pairs shortest path, and its simplicity made it a staple in computer science classes. Frequently, students discover a variant of this Floyd-Warshall algorithm by mixing up the loop order, ending up with the incorrect APSP matrix. This paper considers a computational problem of computing this \textit{incorrect} APSP matrix. We will propose efficient algorithms for this problem and prove that this incorrect variant is APSP-complete.
\end{abstract}

\setcounter{page}{1}

\section{Introduction}
The Floyd-Warshall algorithm computes the all-pairs shortest path (APSP) in a directed weighted graph \cite{floyd, warshall, CLRS}. Known in the computer science community for over 60 years, Floyd-Warshall is still one of the most efficient algorithms for the APSP problem, where it has a runtime of $O(n^3)$ for a graph with $n$ vertices. No algorithms have improved this runtime by a polynomial factor in general graphs, which motivates the APSP Conjecture by \cite{JACM18}. 

Another remarkable characteristic of this algorithm is its simplicity - the standard implementation of this algorithm is a short triply nested loop, as shown below:

\begin{algorithm}
\caption{The Floyd-Warshall Algorithm (KIJ Algorithm)}\label{alg:kij}
\begin{algorithmic}
\State $A \leftarrow $ Adjacency matrix of the graph
\Ensure $A[i, i] = 0$ for all $1 \le i \le n$
\Ensure $G$ has no negative cycles
\For{$k\gets 1, 2, \ldots, n$}
\For{$i\gets 1, 2, \ldots, n$}
\For{$j\gets 1, 2, \ldots, n$}
\State $A[i, j] \leftarrow \min(A[i, j], A[i, k] + A[k, j])$
\EndFor
\EndFor
\EndFor
\end{algorithmic}
\end{algorithm}

Unfortunately, due to the algorithm being too simple, some students write a variant of the Floyd-Warshall algorithm, either by mistake or as a deliberate attempt to rectify the loop order into a \textit{natural} lexicographic order. \footnote{We are unaware of any publication over this variant except \cite{punkapsp}, but some Internet forums and even the lecture material discuss it. Examples are: \\ \scriptsize\url{https://www.quora.com/Why-is-the-order-of-the-loops-in-Floyd-Warshall-algorithm-important-to-its-correctness}, \url{https://cs.stackexchange.com/questions/9636/why-doesnt-the-floyd-warshall-algorithm-work-if-i-put-k-in-the-innermost-loop}, \url{https://codeforces.com/blog/entry/20882}, \url{https://cs.nyu.edu/~siegel/JJ10.pdf}. }
\begin{algorithm}
\caption{The Variant of Floyd-Warshall Algorithm (IJK Algorithm)}\label{alg:ijk}
\begin{algorithmic}
\State $A \leftarrow $ Adjacency matrix of the graph $G$
\Ensure $A[i, i] = 0$ for all $1 \le i \le n$
\Ensure $G$ has no negative cycles
\For{$i\gets 1, 2, \ldots, n$}
\For{$j\gets 1, 2, \ldots, n$}
\For{$k\gets 1, 2, \ldots, n$}
\State $A[i, j] \leftarrow \min(A[i, j], A[i, k] + A[k, j])$
\EndFor
\EndFor
\EndFor
\end{algorithmic}
\end{algorithm}

\newpage
\cref{alg:ijk} do not compute the correct APSP matrix. For example, let $A$ be the adjacency matrix, and the $M_1$ and $M_2$ be the resulting matrix computed by \cref{alg:kij} and \cref{alg:ijk} on $A$. The following choice of $A$ makes $M_1[2, 1]$ and $M_2[2, 1]$ different:

\begin{align*}
A = \begin{bmatrix}
0 & \infty & \infty & \infty \\
\infty & 0 & \infty & 1 \\
1 & \infty & 0 & \infty \\
\infty & \infty & 1 & 0
\end{bmatrix}, 
M_1 = \begin{bmatrix}
0 & \infty & \infty & \infty \\
3 & 0 & 2 & 1 \\
1 & \infty & 0 & \infty \\
2 & \infty & 1 & 0
\end{bmatrix}, M_2 = \begin{bmatrix}
0 & \infty & \infty & \infty \\
\infty & 0 & 2 & 1 \\
1 & \infty & 0 & \infty \\
2 & \infty & 1 & 0
\end{bmatrix}
\end{align*}

Since \cref{alg:ijk} cannot compute the APSP correctly, it is natural to regard it as a novice mistake and move on. But we encounter surprisingly nontrivial questions when we try to understand what's happening. For example, in a sparse graph with no negative edges, we can compute the APSP problem using Dijkstra's algorithm \cite{dijkstra} for computing a single-source shortest path in each vertex. Using the Fibonacci heap from \cite{FT84} yields an $O(nm + n^2 \log n)$ time algorithm for $n$-vertex, $m$-edge graphs. Hence, we have an efficient way to obtain the resulting matrix by \cref{alg:kij}, but for \cref{alg:ijk}, despite being a seemingly \textit{novice} version, it is unclear how to obtain such matrix efficiently. To this end, we formally define the \iapsp problem as follows:

\begin{definition}[\textsc{Incorrect-APSP}]
Given a weighted graph with $n$ vertices without negative cycles, compute the matrix returned by \cref{alg:ijk}.
\end{definition}

We note that calling \cref{alg:ijk} as an \textit{incorrect} algorithm might be unfair. Indeed, a recurring theme in art and fashion is to take a seemingly \textit{correct} version of a piece and twist it cleverly so that people can find something new from a familiar composition. Even in theoretical computer science, such attempts are not new: For example, a famous \textit{Bogosort} \cite{bogosort} is a sorting algorithm deliberately engineered to perform worse. A more similar example where people consider a variant of a well-known algorithm also exists \cite{punkapsp, punksort}. In this way, depending on the inspiration we drew, we can call the \cref{alg:ijk} as either \textsc{Improvised-APSP} (Jazz music reference), \textsc{Punk-APSP} (Rock music reference), or \textsc{Bogo-APSP} (Bogosort reference). However, we will (unfortunately) call the problem \iapsp to avoid possible confusion.

A significant inspiration for this work is an arXiv preprint by \cite{punkapsp}, which considers the exact problem of \iapsp. In the preprint, the authors proved that running \cref{alg:ijk} three times on the given adjacency matrix will compute a correct APSP distance matrix. In other words, the \iapsp problem is not entirely incorrect - even if the student does not know the correct loop order, they can run the algorithm three times and obtain a correct APSP matrix. While this result itself is funny (who cares about the loop order if you can just run three times?), it raises an intriguing question about the hardness of the \iapsp problem, as their result implies that \iapsp is at least as hard as the APSP problem. In other words, a subcubic implementation of \iapsp can imply a breakthrough! So, is this seemingly \textit{novice} mistake of APSP made the problem strictly harder? Or can we find a subcubic reduction to prove that the \iapsp problem is equivalent to the APSP problem?

\subsection{Our Results}
Our results in this paper are the following:

\begin{restatable}{theorem}{sparse}
\label{thm:sparse}
    Given a weighted graph with $n$ vertices and $m$ edges, we can solve the \iapsp in $O(n T_{SSSP}(n, m))$ time, where $T_{SSSP}(n, m)$ is the time to execute a single-source shortest path (SSSP) algorithm in a graph with $n$ vertices and $m$ edges.
\end{restatable}

\begin{restatable}{theorem}{complete}
\label{thm:complete}
\iapsp is subcubic equivalent to \apsp.
\end{restatable}

\cref{thm:sparse} concerns a notable special case where $n$ iteration of single-source shortest path (SSSP) algorithm is known to outperform the Floyd-Warshall algorithm on the APSP problem. For the \iapsp problem, it is unclear whether it can be optimized using the SSSP-based approach. We show this is possible and present an algorithm matching the APSP problem's corresponding bounds. Note that the current record for the SSSP algorithm is $T_{SSSP}(n, m) = O(m \log^2 n \log (n M) \log \log n)$ (\cite{bringmann}). Under the assumption where $G$ do not contain negative weight edges, $T_{SSSP}(n, m) = O(n \log n + m)$ (\cite{FT84}).

\cref{thm:complete} implies that both \iapsp and APSP have a subcubic algorithm or both do not, meaning that we have a new member in the list of APSP-complete problems where the subcubic solution to any of them implies a breakthrough.

\paragraph{Organization of our paper} In \cref{sec:prelim}, we list the definitions and carefully formalize \cref{alg:ijk} in graph theoretic notions. In \cref{sec:mainthm}, we prove that the \iapsp problem is equivalent to the path minimization problem with certain constraints. In \cref{sec:sparse}, we prove \cref{thm:sparse} by solving the path minimization problem by combining single-source shortest path problem and dynamic programming. In \cref{sec:complete}, we prove \cref{thm:complete} by devising an algorithm that solves the \iapsp in subcubic time under the subcubic APSP oracle.

\section{Preliminaries} \label{sec:prelim}
\paragraph{Notations} For integers $i$ and $j$, we use $[i, j]$ to denote the set $\{i, i+1, \ldots, j\}$. Let $G = (V, E, w)$ be a weighted directed graph, where $V = [1, n]$ is the set of vertices, $E \in V \times V$ is a set of edges where an element $(i, j) \in E$ represents a directed edge from vertex $i$ to vertex $j$, and $w: E \rightarrow [-M, M]$ is the weight function of edges. As we represent vertices as integers, we can say that a vertex is \textit{smaller} or \textit{larger} than the other vertex by comparing the integer. We heavily rely on such notation, as our algorithm iterates $V$ in the order of these integers. 

We can use another representation for the directed weighted graphs if we do not rely on the graph's sparsity. For a graph $G = (V, E, w)$ with $V = [1, n]$, an \textit{adjacency matrix} $A(G)$ of $G$ is a $n \times n$ matrix such that:

\begin{align*}
    A(G)[i, j] =\left\{\begin{array}{lr}
        0 & \text{ if } i = j \\
        w(i, j) & \text{ if } i \neq j, (i, j) \in E \\
        \infty & \text{ otherwise }
        \end{array}\right\}
\end{align*}

Note that $A(G)[i, j]$ may not accurately represent $G$ if it has negative-weight loops, but this is not a problem since we will invoke \cref{alg:kij} or \cref{alg:ijk} only if $G$ has no negative cycles. Here, $\infty$ is an element where $\infty + x = \infty$ holds for all $x \in [-M, M]$.

\subsection{Definitions}
\cref{alg:kij} and \cref{alg:ijk} both try to recognize certain kinds of paths in a graph: Starting from the paths with at most one edge, it tries to build up a new paths that is a concatenation of two paths. Here, we will formalize it so that we can view it as a minimization problem over certain types of paths. In \cref{sec:mainthm}, we will argue that there is a simple characterization of paths realized by \cref{alg:ijk}.

We formally define a concept of \textit{path realization}, a generic description of paths realized by the algorithms above.

\begin{definition}\label{def:realize}
A path $P = \{p_0, p_1, \ldots, p_m\}$ is \textbf{realized} by a sequence of 3-tuple \\$T = \{(i_1, j_1, k_1), (i_2, j_2, k_2), \ldots, (i_l, j_l, k_l)\}$, if and only if $m \leq 1$, or there exists a pair of integer $d, x$ such that:
\begin{itemize}
    \item $1 \le d \le l$
    \item $1 \le x \le m - 1$
    \item $(i_d, j_d, k_d) = (p_0, p_m, p_x)$
    \item $\{p_0, p_1, \ldots, p_x\}$ is realized by $\{(i_1, j_1, k_1), (i_2, j_2, k_2), \ldots,(i_{d-1}, j_{d-1}, k_{d-1})\}$.
    \item $\{p_x, p_{x+1}, \ldots, p_m\}$ is realized by $\{(i_1, j_1, k_1), (i_2, j_2, k_2), \ldots, (i_{d-1}, j_{d-1}, k_{d-1})\}$.
\end{itemize}
\end{definition}

We list some immediate corollaries that can be shown by structural induction on \cref{def:realize}. While the corollaries themselves are somewhat trivial, they help us to organize our theorems in a purely mathematical way.

\begin{corollary}\label{def:obvrel_kij}
    Let $T_{kij}(n)$ be a sequence of length $n^3$ consisting of 3-tuples, where the $an^2 + bn + c + 1$-th element in the sequence is $(b+1,c+1,a+1)$ for all $0 \le a, b, c \le n - 1$. Given a $n \times n$ adjacency matrix of $G$, where $G$ has no negative cycle, \cref{alg:kij} will return a matrix $M$, where $M[i, j]$ is the minimum total weight of all simple paths from $i$ to $j$ that is realized by $T_{kij}(n)$.
\end{corollary}
\begin{proof}
    It suffices to prove the statement for all paths, as the graph does not contain any negative cycles, and we can turn any non-simple paths into simple paths without increasing their total length.
    
    For any path from $i$ to $j$ realized by $T(n)$, $M[i, j]$ is at most the weight of such paths. Note that the triple nested loop in \cref{alg:kij} exactly iterates the list $T_{kij}(n)$ and performs a \textit{relaxation} operation of $A[i, j] = \min(A[i, j], A[i, k]+ A[k, j])$ for each 3-tuple $(i, j, k) \in T_{kij}(n)$. Hence, we can show this by induction over \cref{def:realize}. Specifically, we can prove the following: For all path $P$ of length $k$ realized by a sequence of 3-tuples of length $l$, $M[i, j]$ is always at most the weight of $P$ after applying the relaxation operation for first $l$ elements of $T_{ijk}(n)$. Then, we can apply induction over $(k, l)$.

    Conversely, we can also prove that, for any $1 \le i, j \le n$, there exists a path of weight at most $M[i, j]$, which is realized by $T(n)$, which also follows the identical induction as above.
\end{proof}

\begin{corollary}\label{def:obvrel_ijk}
    Let $T_{ijk}(n)$ be a sequence of length $n^3$ consisting of 3-tuples, where the $an^2 + bn + c + 1$-th element in the sequence is $(a+1,b+1,c+1)$ for all $0 \le a, b, c \le n - 1$. Given a $n \times n$ adjacency matrix of $G$, where $G$ has no negative cycle, \cref{alg:ijk} will return a matrix $M$, where $M[i, j]$ is the minimum total weight of all simple paths from $i$ to $j$ that is realized by $T_{ijk}(n)$.
\end{corollary}
\begin{proof}
You can proceed identically as in \cref{def:obvrel_kij}.
\end{proof}

The following provides some characterizations of paths that we will use in the later stage of the paper.

\begin{restatable}{definition}{inc}
\label{def:Pinc}
    A path $P = \{p_0, p_1, \ldots, p_k\}$ with $k \geq 1$ is \textbf{increasing} if for all $1 \le i \le k$ it holds that $p_{i-1} < p_i$.
\end{restatable}
\begin{restatable}{definition}{dec}
\label{def:Pdec}
    A path $P = \{p_0, p_1, \ldots, p_k\}$ with $k \geq 1$ is \textbf{decreasing} if for all $1 \le i \le k$ it holds that $p_{i-1} > p_i$.
\end{restatable}
\begin{restatable}{definition}{valley}
\label{def:Pvalley}
    A path $P = \{p_0, p_1, \ldots, p_k\}$ with $k \geq 1$ is \textbf{valley} if for all $1 \le i \le k - 1$ it holds that $p_i \le \min(p_0, p_k)$.
\end{restatable}
\begin{restatable}{definition}{proper}
\label{def:Pproper}
    A path $P = \{p_0, p_1, \ldots, p_k\}$ with $k \geq 1$ is \textbf{proper} if there is no index $1 \le i \le k - 2$ such that $p_i > \min(p_0, p_k)$ and $p_{i +1} > \min(p_0, p_k)$.
\end{restatable}

\section{Main Theorem}\label{sec:mainthm}
The main theorem of our paper is as follows:
\begin{restatable}{theorem}{main}
\label{thm:main}
In a graph with $n$ vertices, a nonempty simple path $P = \{p_0, p_1,\ldots, p_k\}$ is realized by $T_{ijk}(n)$ if and only if one of the following holds:
\begin{itemize}
    \item $p_0 < p_k$, and there exists an index $0 \le i \le k$ such that $\{p_0, p_1, \ldots, p_i\}$ is proper, $\{p_i, p_{i+1}, \ldots, p_k\}$ is increasing, and $p_i \geq p_0$.
    \item $p_0 > p_k$, and there exists an index $0 \le i \le k$ such that $\{p_0, p_1, \ldots, p_i\}$ is decreasing, $\{p_i, p_{i+1}, \ldots, p_k\}$ is proper, and $p_i \geq p_k$.
\end{itemize}
\end{restatable}

By combining \cref{thm:main} with \cref{def:obvrel_ijk}, we can yield an explicit characterization for the output of \cref{alg:ijk}.

\begin{corollary}\label{cor:main}
    Given a $n \times n$ adjacency matrix of $G$ where $G$ has no negative cycle, \cref{alg:ijk} will return a matrix $M$, where $M[i, j]$ is:
    \begin{itemize}
        \item If $i < j$, the minimum possible total weight of a path $P = \{p_0 = i, p_1, \ldots, p_k = j\}$ such that there exists an index $0 \le x \le k$ such that $\{p_0, p_1, \ldots, p_x\}$ is proper, $\{p_x, p_{x+1}, \ldots, p_k\}$ is increasing, and $p_x \geq p_0$.
        \item If $i = j$, $0$.
        \item if $i > j$, the minimum possible total weight of a path $P = \{p_0 = i, p_1, \ldots, p_k = j\}$ such that $\{p_0, p_1, \ldots, p_x\}$ is decreasing, $\{p_x, p_{x+1}, \ldots, p_k\}$ is proper, and $p_x \geq p_0$.
    \end{itemize}
\end{corollary}
\begin{proof}
    By \cref{def:obvrel_ijk}, \cref{alg:ijk} returns a matrix where $M[i, j]$ is the minimum total weight of all simple paths from $i$ to $j$ realized by $T_{ijk}(n)$, which is of above form by \cref{thm:main}.
\end{proof}

\subsection{High-level ideas}

Before proceeding to the proof of \cref{thm:main}, let's play with several examples to motivate intuition. Recall the standard proof of \cref{alg:kij}, where one proves the following lemma by the induction on $t$:

\begin{lemma}[Key lemma of \cite{floyd}]\label{lem:floydwarshall}
    For all $0 \le t \le n$, a path $P = \{p_0, p_1, \ldots, p_k\}$ is realized by the $tn^2$-length prefix of $T_{kij}(n)$ if and only if $p_i \le t$ for all $1 \le i \le k - 1$.
\end{lemma}

This proof strategy surely does not work for \cref{alg:ijk}, but it does work identically for the specific type of path, which we define as a valley path:

\valley*

However, while all valley paths are realized by $T_{ijk}(n)$, this is still not an exhaustive characterization. Consider the path $P = \{3, 101, 1, 102, 2\}$, which is not a valley path but is realized by $T_{ijk}(n)$ as the subsequence $\{(1, 2, 102), (3, 1, 101), (3, 2, 1)\}$ of $T_{ijk}(n)$ can realize $P$. The reason is that the algorithm can glue some large vertices into $P$ before it starts to realize the valley path. For this specific counterexample, one can observe that these large vertices should be adjacent to small vertices to get glued inside a valley path, which motivates the definition of a proper path.

\proper*

With some modification of proofs in \cref{lem:floydwarshall}, we can see that all proper paths are realized by $T_{ijk}(n)$. However, this definition is slightly strict: For example, the path $P = \{1, 2, 3, 4\}$ is realized by $\{(1, 3, 2), (1, 4, 3)\}$ but does not fit in a definition of valley path. More generally, we can append any increasing paths in the back of the proper path and prepend any decreasing paths in the front of the proper path. By addressing these cases, we do reach an appropriate characterization of paths realized by $T_{ijk}(n)$, that fits nicely to the inductive argument given in the proof of \cref{thm:main}.

\subsection{Proof of the Main Theorem}
\main*

\begin{proof}[Proof ($\leftarrow$)]
    Let $n$ be the number of vertices, and let $T_{ijk}(n, i, j)$ be the prefix of $T_{ijk}(n)$ up to and including the element $(i, j, n)$. 
    
    We first show that all proper paths from $i$ to $j$ are realized by $T_{ijk}(n, i, j)$. We perform an induction over $(i, j)$. Consider a proper path $P = \{p_0, p_1, \ldots, p_k\}$, and assume that there is no index $1 \le i \le k - 1$ such that $p_i \le \min(p_0, p_k)$. Then, since there could be no two entries with $p_i > \min(p_0, p_k)$, it holds that $k \le 2$, and we are done. Otherwise, take the maximum $p_x$ such that $p_x < \min(p_0, p_k)$. By maximality, both $\{p_0, p_1, \ldots, p_x\}$ and $\{p_x, p_{x+1}, \ldots, p_k\}$ are proper paths, and by inductive hypothesis, they are realized in $T_{ijk}(p_0, p_x)$ and $T_{ijk}(p_x, p_k)$. Hence, by \cref{def:realize}, $P$ is realized. 

    As $P = \{p_0, p_1, \ldots, p_k\}$ from $p_0$ to $p_k$ are realized by $T_{ijk}(n, p_0, p_k)$, we can see $\{p_0, p_1, \ldots, p_k, x\}$ are realized by $T_{ijk}(n, p_0, x)$ for any $x > p_k$, and similarly, $\{x, p_0, p_1, \ldots, p_k\}$ are realized by $T_{ijk}(n, x, p_k)$ for any $x > p_0$. Using this fact, we can use a similar induction to prove that we can append an increasing path in the back or a decreasing path in the front.
\end{proof}
\begin{proof}[Proof ($\rightarrow$)]
    We use the induction on the length $l$ of the prefix of $T_{ijk}(n)$ to show that the path realized by a length-$l$ prefix of $T(n)$ is always in one of such patterns. This claim is true for $l = 0$. Assume $l \geq 1$ and let $(i, j, k)$ be the last element of $T(n)$. Note that we can assume $i \neq j, j \neq k, k \neq i$, since if one of them holds, we have no new realized paths or a non-simple path. 
    
    We list all possible cases and show we can always find such patterns. Here, we use $*$ to denote that any index from $1$ to $n$ can be there.
    \begin{itemize}
        \item $i < j$, $k < i$: We consider a path in a form of $i \rightarrow \text{(decreasing)} \rightarrow \text{(proper)} \rightarrow k \rightarrow \text{(proper)} \rightarrow \text{(increasing)} \rightarrow j$. Let $z$ be the first vertex in the increasing part of this path, where $i \leq z$ (as $i < j$ such vertex exists). The path between $i$ to $z$ is proper, and the path between $z$ to $j$ is increasing. 
        \item $i < j$, $i < k < j$: All entries of $(k, j, *)$ are not in the current prefix, so the only path from $k$ to $j$ realized now is the trivial path $\{k, j\}$. We can append this in the increasing part of our path. 
        \item $i < j$, $j < k$: All entries of $(i, k, *)$ and $(k, j, *)$ are not in the current prefix, so the path we consider is exactly $\{i, k, j\}$ which is proper.  
        \item $j < i$, $k < j$: We consider a path in a form of $i \rightarrow \text{(decreasing)} \rightarrow \text{(proper)} \rightarrow k \rightarrow \text{(proper)} \rightarrow \text{(increasing)} \rightarrow j$. Let $z$ be the last vertex in the decreasing part of this path, where $j \leq z$ (as $j < i$ such vertex exists). The path between $i$ to $z$ is decreasing, and the path between $z$ to $j$ is proper. 
        \item $j < i$, $j < k < i$: All entries of $(i, k, *)$ are not in the current prefix, so the only path from $i$ to $k$ realized now is the trivial path $\{i, k\}$. We can prepend this in the decreasing part of our path. 
        \item $j < i$, $i < k$: All entries of $(i, k, *)$ and $(k, j, *)$ are not in the current prefix, so the path we consider is exactly $\{i, k, j\}$ which is proper.  
    \end{itemize}
\end{proof}

\section{Algorithm for \iapsp Problem} \label{sec:sparse}
In this section, we prove the following theorem:

\sparse*

For this, we will try to solve the optimization problem described in \cref{cor:main}. Here, it is useful to observe the following:

\begin{observation}\label{obs:reverse}
Let $rev(P)$ be the reverse of the path. A path $P$ is realized by $T_{ijk}(n)$ if and only if $rev(P)$ is realized by $T_{ijk}(n)$.\end{observation}
\begin{proof}
    If $P$ satisfies the property from \cref{thm:main}, the reverse also satisfies the property from \cref{thm:main}.
\end{proof}

By \cref{obs:reverse}, it suffices to find $M[i, j]$ for all $i < j$, as the other case can be solved by reversing all edges and repeating the same algorithm. 

We will fix $i$ and devise an algorithm that computes $M[i, j]$ for all $j \geq i$. To start, we find a minimum-length proper path to all $j$. By the description in \cref{cor:main}, we only need to compute the proper path ending in $j \geq i$. By \cref{def:Pproper}, a path is proper if it does not contain two adjacent vertices $p_x, p_{x+1}$ of index greater than $\min(i, j)$ unless one of the vertices is the last vertex of the path. 

Since we have $\min(i, j) = i$, this is equivalent to not using an edge where both endpoints have an index greater than $i$, except where the edge is the last on the path. As we fixed $i$, we can apply SSSP algorithm on $G$ where edges with endpoints greater than $i$ are removed. Let $\textsc{ProperExceptLast}[j]$ be the distance from $i$ to $j$ on such graph. Then, for the last move, we simulate the move for all vertex. Specifically, we initialize $\textsc{Proper}[j] = \textsc{ProperExceptLast}[j]$ for all $1 \le j \le n$, and for all edge $(u, v) \in E$, we set $\textsc{Proper}[v] = \min(\textsc{Proper}[v], \textsc{ProperExceptLast}[u] + w(u \rightarrow v))$. In the end, $\textsc{Proper}[j]$ holds the shortest proper path from $i$ to $j$ for all $j \geq i$.

Finally, we may need to append the increasing path at the back of each proper path. We can use dynamic programming as the increasing path consists of edges headed to the larger indexed vertex. For computing the $M[i, j]$, we have two choices:
\begin{itemize}
    \item Assume $p_x = p_k$ and simply take the shortest proper path from $i$ to $j$.
    \item Otherwise, we pick an edge $(k \rightarrow j)$ and use the edge as a final edge in the optimal path. Here, $k < j$ should hold, and there should be an edge $(k, j) \in E$. The remaining path from $i$ to $k$ is a proper path appended with an increasing path computed by $M[i, k]$.
\end{itemize}

In \cref{alg:sparse}, we present a pseudocode of the above algorithm.

\begin{algorithm}\caption{Computing $M[i, j]$ for fixed $i$, all $j \geq i$}\label{alg:sparse}
\begin{algorithmic}
\State $G^\prime = (G$ without edges where both vertices have index greater than $i)$
\State $\textsc{Proper} = \textsc{SSSP}(G^\prime, i)$
\State $\textsc{ProperExceptLast} = \textsc{Proper}$
\For{$(u, v) \in E$}
\State $\textsc{Proper}[v] = \min(\textsc{Proper}[v], \textsc{ProperExceptLast}[u] + w(u \rightarrow v))$.
\EndFor
\For{$j\gets i, i+1, \ldots, n$}
\State $M[i, j] = \textsc{Proper}[j]$
\For{$k$ incident to $j$}
\If{$i \le k < j$}
\State $M[i, j] = \min(M[i, j], M[i, k] + w(k \rightarrow j))$
\EndIf
\EndFor
\EndFor
\end{algorithmic}
\end{algorithm}

As all procedures except the SSSP use $O(n + m)$ computation, the algorithm's running time is dominated by the SSSP, which leads to the runtime of $O(nT_{SSSP}(n, m))$. 

\section{Subcubic equivalence between APSP and \iapsp} \label{sec:complete}
In this section, we will prove \cref{thm:complete}. For a computational problem $A, B$, we say $A \leq_3 B$ if there is a subcubic reduction from $A$ to $B$ as in \cite{JACM18}. The following result is known:

\begin{restatable}[Theorem 1 of \cite{punkapsp}]{theorem}{hard}
\label{thm:hard}
$\textsc{APSP} \leq_3 \textsc{Incorrect-APSP}$.
\end{restatable}

In the first subsection, we will recite the proof of \cref{thm:hard} from \cite{punkapsp}. We emphasize that this is not our original contribution: Our goal is to make this paper self-contained and to rephrase the statement of \cite{punkapsp} in terms of subcubic reduction.

In the second subsection, we complement \cref{thm:hard} by proving the following theorem, hence closing the gap:

\begin{restatable}{theorem}{easy}
\label{thm:easy}
$\textsc{Incorrect-APSP} \leq_3 \textsc{APSP}$.
\end{restatable}

Given that \cref{thm:easy} and \cref{thm:hard} are true, the proof of \cref{thm:complete} follows by definition of subcubic equivalence in \cite{JACM18}.
\complete*

\subsection{Proof of \cref{thm:hard}}
\begin{lemma}\label{lem:hard}
    Let $T^i$ be a sequence $T$ repeated for $i$ times and $P$ be any simple path in a graph with $n$ vertices. $(T_{ijk}(n))^3$ realizes $P$.
\end{lemma}
\begin{proof}
    Given a path $P = \{p_0, p_1,\ldots, p_k\}$, let $m$ be an integer in $[0, k]$ where $p_m = \max_{i = 0}^{k} p_i$. We call $p_i$ a \textit{skyscraper}, if it satisfies the following:
    \begin{itemize}
        \item $i < m$, and there exists no $j < i$ such that $p_j > p_i$.
        \item $i = m$.
        \item $i > m$, and there exists no $j > i$ such that $p_j > p_i$.
    \end{itemize}
    It is helpful to observe that $p_0$ and $p_k$ are always a skyscraper.
    
    Let $0 = i_0 < i_1 < \ldots < i_j = m < i_{j+1} < \ldots < i_l = k$ be a list of indices where $p_{i_s}$ is a skyscraper for all $0 \le s \le l$. For all $1 \le s \le l$, it holds that the subpath $p_{i_{s-1}}, p_{i_{s-1} + 1}, \ldots, p_{i_s}$ is proper - in fact, there can't be even a single vertex with value larger than $\min(p_{i_{s-1}}, p_{i_s})$ since that will add such vertex into a list of skyscraper, contradicting our assumption that $i_{s-1}$ and $i_s$ are adjacent. By \cref{thm:main}, we can see that the subpath $p_{i_{s-1}}, p_{i_{s-1} + 1}, \ldots, p_{i_s}$ is realized by $(T_{ijk}(n))^1$.

    Next, we will prove that both subpath $p_0, p_1, \ldots, p_m$ and $p_m, p_{m+1}, \ldots, p_k$ are realized in $(T_{ijk}(n))^2$. Given that all subpaths between $p_{i_{s-1}}$ and $p_{i_s}$ are realized in $(T_{ijk}(n))^1$, the remaining skyscrapers are either increasing or decreasing according to its relative location per $p_m$. The statement holds as we can realize the increasing and decreasing path by \cref{thm:main}.

    Finally, to realize $P$ in $(T_{ijk}(n))^3$, we need $(p_0, p_k, p_m)$ in $T_{ijk}(n)$ which we definitely have.
\end{proof}
\hard*
\begin{proof}
    By Definition 3.1 in \cite{JACM18}, it suffices to design a subcubic algorithm for $\textsc{APSP}$ which calls the oracle to compute $\textsc{Incorrect-APSP}$ for an $n \times n$ matrix, at most a polylogarithmic times. 
    
    We take an input graph, construct an adjacency matrix, call $\textsc{Incorrect-APSP}$, call $\textsc{Incorrect-APSP}$ again in the returned matrix, call $\textsc{Incorrect-APSP}$ again in the returned matrix, and return the output of $\textsc{Incorrect-APSP}$. This algorithm calls an oracle for $3$ times and runs in quadratic time, which satisfies all requirements to be a subcubic reduction. The correctness follows by \cref{lem:hard}.
\end{proof}

\subsection{Proof of \cref{thm:easy}}
Let $A \odot B$ be a min-plus matrix product of two $n \times n$ matrix $A$ and $B$. The following lemma shows that we can convert a proper path minimization into a valley path minimization problem. Note that the definition of $G^2$ is equivalent to the usual definition of graph powers, defined as a power of adjacency matrix (here, the multiplication operator is $\odot$).

\begin{lemma}  \label{lem:propvalley}
Given a graph $G$, A proper path of $G$ from $i$ to $j$ with weight $w$ exists if and only if a valley path of $G^2$ from $i$ to $j$ with weight $w$ exists. Here, $G^2$ is a complete directed graph on the same set of vertices with $G$, where the weight of the edge from $i$ to $j$ is the minimum weight path from $i$ to $j$ using at most $2$ edges.
\end{lemma}
\begin{proof}
    Every proper path of $G$ translates to a valley path of $G^2$, as all vertices that are not $i, j$ and have an index at most $\min(i, j)$ are either adjacent in the path or have at most one intermediate vertex in the path. Conversely, given a valley path of $G^2$, we can turn it into a path in $G$ by replacing an edge with two edges and a vertex. Those new are the only vertex that can violate the $p_i \le \min(p_0, p_k)$ condition, but they are not adjacent by construction.
\end{proof}

We will devise an algorithm that uses subcubic oracle for APSP-complete problems to compute the resulting matrix $M[i, j]$. Note that we will only demonstrate how to compute $M[i, j]$ for $i \leq j$ - as in \cref{obs:reverse}, we can compute the opposite part by computing the transpose of $A(G)$ in quadratic time and running the same algorithm. We use the result from \cite{JACM18} that APSP and min-plus matrix multiplication are subcubic equivalent. 

Let's first find a matrix for minimum proper paths in $G$, which by \cref{lem:propvalley} is equivalent to minimum valley paths in $G^2$. Let $V_l[i, j]$ be a minimum weight of the valley path from $i$ to $j$, which includes all valley paths with at most $2^l$ edges. Note that this definition includes all valley paths with $\leq 2^l$ edges, but it is not limited to such valley paths - it is, however, limited to valley paths. Here, $V_0 = A(G^2) = A(G)^2$, and we can compute this using min-plus matrix multiplication.

To compute the entry $V_l[i, j]$, we fix the vertex $k$ in the middle of the (imaginary) path $P$. As $k$ is in the middle of $P$, and since $P$ has at most $2^l$ edges, the path to the left and the right of $k$ has at most $2^{l-1}$ edges. 

Let's say a valley path from $i$ to $j$ is \textit{ascending} if $i < j$ and \textit{descending} otherwise. The path to the left of $k$ can be divided into a block of descending paths, and the path to the right of $k$ can be divided into a block of ascending paths. To see this, consider all vertex $v$ such that the subpath from $v$ to $k$ does not have any vertex greater than $v$, which we refer to as a \textit{skyscraper}. Then, for each adjacent skyscraper, the subpath they form is a valley, as otherwise, we find a skyscraper in between them.

Conversely, it's easy to see that a concatenation of ascending valley paths, followed by descending valley paths, forms a valley path. Hence, a valley path of at most $2^l$ edges is equivalent to a sequence of descending and ascending valley paths, each with at most $2^{l-1}$ edges.

The minimum path from $i$ to $j$ that is a sequence of descending valley paths of at most $2^{l-1}$ edges, can be computed with the APSP oracle: We can provide an adjacency matrix of $V_{l-1}[i, j]$, where all entries with $i < j$ are overwritten to $\infty$. Conversely, we can compute the sequence of ascending paths with the APSP oracle by overwriting $V_{l-1}[i, j]$ with $ I> j$ to $\infty$. Then, we can combine those two patterns by a single min-plus matrix multiplication.

Finally, we need to append an increasing path in the back of the path. By supplementing the adjacency matrix of $G$ where we overwrite all entries with $i > j$ to $\infty$, we can compute the increasing path of minimum cost for all pairs using an APSP oracle. We can obtain the desired answer by multiplying this with the valley path matrix.

In \cref{alg:complete}, we present a pseudocode of the above algorithm.

\begin{algorithm}\caption{Computing $M[i, j]$ for all $1 \le i \le j \le n$}\label{alg:complete}
\begin{algorithmic}
\Ensure $A \odot B$ is subcubic
\Ensure $APSP(G)$ is subcubic
\State $V_0 = A(G) \odot A(G)$
\For{$l\gets1, 2, \ldots, \lceil \log_2 n \rceil$}
\State $\textsc{DescendingValley} = V_{l-1}$
\State $\textsc{AscendingValley} = V_{l-1}$

\For{$i\gets1, 2, \ldots, n$}
\For{$j\gets1, 2, \ldots, n$}
\If{$i < j$}
\State $\textsc{DescendingValley}[i, j] = \infty$
\ElsIf{$i > j$}
\State $\textsc{AscendingValley}[i, j] = \infty$
\EndIf
\EndFor
\EndFor
\State $V_l = APSP(\textsc{DescendingValley}) \odot APSP(\textsc{AscendingValley})$
\EndFor
\State $\textsc{Valley} = V_{\lceil \log_2 n \rceil}$
\State $G^\prime = (G \text{ with edges $u \rightarrow v$ such that }u < v)$
\State $\textsc{Answer} = \textsc{Valley} \odot APSP(A(G^\prime))$
\For{$i\gets1, 2, \ldots, n$}
\For{$j\gets1, 2, \ldots, n$}
\If{$i \leq j$}
\State $M[i, j] = \textsc{Answer}[i, j]$
\EndIf
\EndFor
\EndFor
\end{algorithmic}
\end{algorithm}

As we make at most $O(\log n)$ calls to the subcubic oracles, \cref{alg:complete} is subcubic, which proves \cref{thm:easy}.

\bibliographystyle{alpha}
\bibliography{ref}

\begin{thebibliography}{CLRS22}

\bibitem[BCF23]{bringmann}
K.~Bringmann, A.~Cassis, and N.~Fischer.
\newblock Negative-weight single-source shortest paths in near-linear time: Now faster!
\newblock In {\em 2023 IEEE 64th Annual Symposium on Foundations of Computer Science (FOCS)}, pages 515--538, Los Alamitos, CA, USA, nov 2023. IEEE Computer Society.

\bibitem[CLRS22]{CLRS}
Thomas~H Cormen, Charles~E Leiserson, Ronald~L Rivest, and Clifford Stein.
\newblock {\em Introduction to algorithms}.
\newblock MIT press, 2022.

\bibitem[Dij22]{dijkstra}
Edsger~W Dijkstra.
\newblock A note on two problems in connexion with graphs.
\newblock In {\em Edsger Wybe Dijkstra: His Life, Work, and Legacy}, pages 287--290. 2022.

\bibitem[Flo62]{floyd}
Robert~W Floyd.
\newblock Algorithm 97: shortest path.
\newblock {\em Communications of the ACM}, 5(6):345, 1962.

\bibitem[FT87]{FT84}
Michael~L. Fredman and Robert~Endre Tarjan.
\newblock Fibonacci heaps and their uses in improved network optimization algorithms.
\newblock {\em J. ACM}, 34(3):596–615, jul 1987.

\bibitem[Fun21]{punksort}
Stanley P.~Y. Fung.
\newblock Is this the simplest (and most surprising) sorting algorithm ever?, 2021.

\bibitem[HKM19]{punkapsp}
Ikumi Hide, Soh Kumabe, and Takanori Maehara.
\newblock Incorrect implementations of the floyd--warshall algorithm give correct solutions after three repeats, 2019.

\bibitem[Ray96]{bogosort}
Eric~S Raymond.
\newblock {\em The new hacker's dictionary}.
\newblock Mit Press, 1996.

\bibitem[War62]{warshall}
Stephen Warshall.
\newblock A theorem on boolean matrices.
\newblock {\em J. ACM}, 9(1):11–12, jan 1962.

\bibitem[WW18]{JACM18}
Virginia~Vassilevska Williams and R.~Ryan Williams.
\newblock Subcubic equivalences between path, matrix, and triangle problems.
\newblock {\em J. ACM}, 65(5), aug 2018.

\end{thebibliography}

\end{document}